%% file: main.tex
\documentclass[conference]{IEEEtran}
\usepackage{cite}
\usepackage{graphicx}
\graphicspath{ {figures/} }
\newcommand{\remove}[1]{}
\usepackage{times}
\usepackage{amsthm}
\usepackage{amssymb}
\usepackage{amsmath}
\usepackage{comment}
\usepackage{dblfloatfix}

\usepackage{datetime}
\usepackage{url}
\usepackage{hyperref}

\usepackage{tabularx}
\newcolumntype{C}[1]{>{\centering\let\newline\\\arraybackslash\hspace{0pt}}m{#1}}

\usepackage{algpseudocode}
\usepackage{algorithm}
\usepackage{bm} 
\usepackage{xcolor} 
\algdef{SE}[DOWHILE]{Do}{doWhile}{\algorithmicdo}[1]{\algorithmicwhile\ #1}%

\newcommand{\pvs}{\vspace{-8pt}}
\newcommand{\param}[1]{{\bf$\langle$#1$\rangle$}} 
\newcommand{\lam}{\gets}
\newcommand{\mib}[1]{$\bm{#1}$}
\newcommand{\algsize}{\scriptsize} 
\renewcommand{\Comment}[1] {\hfill\textit{\textcolor{cerulean}{$\triangleright$~#1}}} 

\algdef{SE}[EVENT]{Event}{EndEvent}[1]{\textbf{upon}\ #1\ \algorithmicdo}{\algorithmicend\ \textbf{event}}%
\algtext*{EndEvent}

\algdef{SE}[EVERY]{Every}{EndEvery}[1]{\textbf{every}\ #1\ \algorithmicdo}{\algorithmicend\ \textbf{every}}%
\algtext*{EndEvery}

\makeatletter
\algnewcommand{\OnlyComment}[1] {\hskip\ALG@thistlm\textit{\textcolor{cerulean}{$\triangleright$~#1}}}
\algnewcommand{\LineIndentComment}[1] {\Statex \hskip\ALG@thistlm\hskip\algorithmicindent\textit{\textcolor{cerulean}{$\triangleright$~#1}}} 
\algnewcommand{\LineComment}[1] {\Statex \hskip\ALG@thistlm\textit{\textcolor{cerulean}{$\triangleright$~#1}}} 
\makeatother

\algnewcommand\algorithmicforeach{\textbf{for each}}
\algdef{S}[FOR]{ForEach}[1]{\algorithmicforeach\ #1\ \algorithmicdo}

\algnewcommand{\EndIIf}{\unskip\ \algorithmicend\ \algorithmicif}
\definecolor{cerulean}{rgb}{0.0, 0.48, 0.65}
\usepackage{array}
\usepackage{calc}

\usepackage{pifont}
\newcommand{\xmark}{\ding{55}}%

\usepackage{multirow}

\usepackage{cancel}
\usepackage{graphicx}
\usepackage{caption}
\DeclareCaptionType{copyrightbox}
\usepackage{subcaption}

\newtheorem{lemma}{Lemma}

\newtheorem{proposition}{Proposition}

\newcommand{\ts}{PaRiS}

\renewcommand{\baselinestretch}{0.99}

\date{}

\author{\IEEEauthorblockN{Kristina Spirovska$^1$, Diego Didona$^1$, Willy Zwaenepoel$^{1,2}$\\
$^1$ EPFL, $^2$ University of Sydney}

}

\begin{document}


\title{\ts{}: Causally Consistent Transactions with Non-blocking Reads and Partial Replication}

\maketitle
\thispagestyle{plain}
\pagestyle{plain}

\begin{abstract}

Geo-replicated data platforms are at the backbone of several large-scale online services. Transactional Causal Consistency (TCC) is an attractive consistency level for building such platforms. TCC avoids many anomalies of eventual consistency, eschews the synchronization costs of strong consistency, and supports interactive read-write transactions.  Partial replication is another attractive design choice for building geo-replicated platforms, as it increases the storage capacity and reduces update propagation costs.

This paper presents \ts{}, the first TCC system that supports partial replication and implements non-blocking parallel read operations, whose latency is paramount for the performance of read-intensive applications. \ts{} relies on a novel protocol to track dependencies, called Universal Stable Time (UST). By means of a lightweight background gossip process, UST identifies a snapshot of the data that has been installed by every DC in the system. Hence, transactions can consistently read from such a snapshot on any server in any replication site without having to block. Moreover, \ts{} requires only one timestamp to track dependencies and define transactional snapshots, thereby achieving resource efficiency and scalability.

We evaluate \ts{} on a large-scale AWS deployment composed of up to 10 replication sites. We show that \ts{} scales well with the number of DCs and partitions, while being able to handle larger data-sets than existing solutions that assume full replication. We also demonstrate a performance gain of non-blocking reads vs. a blocking alternative (up to 1.47x higher throughput with 5.91x lower latency for read-dominated workloads and up to 1.46x higher throughput with 20.56x lower latency for write-heavy workloads).

\end{abstract}


\input{sections/introduction.tex}

\input{sections/system_model.tex}

\input{sections/design.tex}

\input{sections/protocols.tex}

\input{sections/evaluation.tex}

\input{sections/related_work.tex}

\input{sections/conclusion.tex}

\section*{Acknowledgments}
This research has been supported by The Swiss National Science Foundation through Grant No. 166306 and by an EcoCloud postdoctoral research fellowship.

\bibliographystyle{IEEEtran}
\bibliography{biblio}

\end{document}

%% file: sections/introduction.tex
\section{Introduction}
\label{sec:intro}
Modern large-scale data platforms rely on geo-replication and sharding to store and manipulate large volumes of data efficiently. Geo-replication allows keeping a copy of the data in a data center (DC) closer to the users, thus reducing access latencies. Sharding enables horizontal scalability, by slicing the dataset in disjoint partitions, each of which can be assigned to a different server.

In geo-replicated environments, partial replication is an effective technique to improve storage capacity and reduce replication costs. In partial replication, each  DC stores only a subset of the partitions. Hence, the system can scale to a higher number of partitions with respect to a full replication approach, and updates performed in one DC are propagated to fewer replicas. 

Causal Consistency (CC) has emerged recently as an attractive consistency model for geo-replicated data platforms~\cite{Lloyd:2011,Lloyd:2013,Zawirski:2015,Akkoorath:2016,Mehdi:2017,Du:2014,Du:2013,Almeida:2013,Spirovska:2017,Didona:2018,Roohitavaf:2017}. CC provides intuitive semantics and avoids many anomalies allowed by weaker models, such as eventual consistency~\cite{DeCandia:2007}. Moreover, CC eschews the high synchronization costs of stronger consistency levels, such as linearizability~\cite{Herlihy:1990}. Transactional CC (TCC)~\cite{Zawirski:2015,Akkoorath:2016} extends CC by providing transactions that observe a causally consistent view of the data, and can perform atomic multi-object writes.

\pvs
~\\\noindent{\bf \ts{}.} This paper presents \ts{}, the first system that implements TCC in a partially replicated data platform, and that supports non-blocking parallel read operations (and hence, non-blocking one-round read-only transactions). Parallel non-blocking reads are an important requirement to guarantee good performance~\cite{Lu:2016,Tomsic:2018,Corbett:2013}, especially for the important and wide class of read-intensive applications~\cite{Urdaneta:2009,Nishtala:2013,Noghabi:2016}.

Achieving non-blocking  parallel transactional reads with partial replication is challenging. This is mainly because different reads within the same transaction may be served in parallel by servers in different DCs. In existing approaches, a DC is not aware of the set of transactions performed by other DCs. Ultimately, this can lead to consistency violations because a server in a DC is unaware of which data is returned  by other servers in other DCs to the same transaction.

\ts{} addresses this issue by means of a new causal dependency tracking protocol, that we call Universal Stable Time (UST). In short, UST identifies a snapshot of the dataset that has been installed in {\em all} DCs. Hence, a transaction can read from such snapshot in any DC without blocking. 
In addition to the snapshot defined by UST, \ts{} equips clients with a private cache, in which clients store their own updates that are not reflected yet in the snapshot identified by the UST. This allows \ts{} to achieve TCC even if exposing to clients a snapshot of the data store that is slightly in the past. 
\ts{} implements UST efficiently as a periodic, lightweight intra- and inter-DC gossiping protocol. In addition, \ts{} uses only one timestamp to track dependencies and to define transactional snapshots, thus enabling scalability both in terms of number of DCs and of partitions per DC. 

The trade-off --which is provably unavoidable~\cite{Tomsic:2018}-- made by \ts{} is to expose to transactions a view of the data that is slightly in the past. We argue that a moderate increase in the data staleness is a reasonable price to pay for the performance benefits brought by \ts{}.  

Overall, \ts{} achieves a triad of low latency, high storage capacity and rich transactional semantics. This represents a significant improvement over existing systems that either do not support partial replication, or generic read-write transactions, or block read operations to preserve consistency.

We evaluate \ts{} on a large scale AWS deployment comprising of up to 10 DCs, and with heterogeneous workloads with different degrees of locality in  the data accesses. We compare \ts{} with a variant of \ts{} that supports partial replication by blocking read operations. We show that \ts{} scales well with the number of DCs and partitions, while being able to handle larger datasets than existing solutions that assume full replication.

\pvs
~\\\noindent{\bf Roadmap.} The remainder of the paper is organized as follows. Section~\ref{sec:model} presents the system model. Section~\ref{sec:design} describes the design of \ts{}. Section~\ref{sec:protocols} describes the protocols and correctness of \ts{}. Section~\ref{sec:eval} reports the results of the evaluation of \ts{}. Section~\ref{sec:relwork} discusses related work. Section~\ref{sec:conclusion} concludes the paper.

%% file: sections/system_model.tex
\section{Definitions and system Model}
\label{sec:model}

\subsection{Causal Consistency}
\label{sec:model:cc}
A system is causally consistent if its servers return values that are consistent with the order defined by the {\em causality} relationship.
Causality is defined as a happens-before relationship between two events~\cite{Ahamad:1995,Lamport:1978}. For two operations $a,\  b$, we say that $b$ causally depends on $a$, and write $a \leadsto b$, if and only if at least one of the following conditions holds: $i)$ $a$ and $b$ are operations in a single thread of execution, and $a$ happens before $b$; $ii)$ $a$ is a write operation, $b$ is a read operation, and $b$ reads the version written by $a$; $iii)$ there is some other operation $c$ such that $a \leadsto c$ and $c \leadsto b$. Intuitively, CC ensures that if a client has seen the effects of operation $b$ and $a\leadsto b$, then the client also sees the effects of operation $a$.

We use lower case letters, e.g., $x$, to refer to a key and the corresponding capital letter, e.g., $X$ to refer to a version of the key. We say that $X$ depends on $Y$ if the write of $X$ causally depends on the write of $Y$. 

\subsection{Transactional Causal Consistency} 
\label{sec:background:tcc}
\noindent{\bf Semantics.} TCC extends CC by means of interactive read-write transactions in which clients can read and write multiple items. TCC enforces two properties.

\pvs
~\\\noindent{\em 1. Read from a causal snapshot.} A causally consistent snapshot is a set of item versions such that all causal dependencies of those versions are also included in the snapshot. All transactions read from a causally consistent snapshot. For any two items, $x$ and $y$, if $X \leadsto Y$ and both $X$ and $Y$ belong to the same causally consistent snapshot, then there is no other $X'$, such that $X'$ is created after $X$ and $X\leadsto X' \leadsto Y$.

\pvs
~\\
Transactional reads from a causal snapshot prevent undesirable anomalies which can arise by simply issuing multiple consecutive single read operations ~\cite{Lloyd:2011}.

The majority of existing CC systems implement transactional reads by means of one-shot read-only transactions~\cite{Du:2014,Lloyd:2011,Lloyd:2013,Lu:2016}.

\pvs
~\\\noindent{\em 2. Atomic updates.} Either all the items written by a transaction are visible to other transactions, or none is. If a transaction writes $X$ and $Y$, then any snapshot visible to other transactions either includes both $X$ and $Y$, or none of them.

 \pvs
 ~\\\noindent{\bf Conflict resolution.} Two writes are conflicting if they are not related by causality  and  update the same key. Conflicting writes are resolved by means of a commutative and associative function, that decides the value corresponding to a key given its current value and the set of updates on the key~\cite{Lloyd:2011}. 

 For simplicity, \ts{} resolves write conflicts using the last-writer-wins rule~\cite{Thomas:1979} based on the timestamp of the updates. Possible ties are settled by looking at the id of the DC combined with the identifier of the transaction that created the update. \ts{} can be extended to support other conflict resolution mechanisms~\cite{Akkoorath:2016,Lloyd:2011,Lloyd:2013,Shapiro:2011}.

\subsection{System model}
\label{sec:model:model}
We assume a distributed key-value store whose dataset is split into $N$ partitions. Each key is deterministically assigned to one partition by a hash function. We assume that each server is assigned a single partition and we note $p_x$ the server responsible for key $x$. 
Each partition $p_i$ is replicated at $R$ different DCs, where $R$ is the replication factor of data. There are $M$ DCs in total, hence, there is only a fragment of the full dataset present in each DC. 

We assume a multi-master system, i.e., all replicas can update keys they are responsible for. Updates are replicated asynchronously to remote DCs. 

We assume a multi-version data store. An update operation creates a new version of a key. Each version stores the value corresponding to the key and some meta-data  to track causality. The system periodically garbage-collects old versions of keys. 
Partitions communicate through point-to-point lossless FIFO channels (e.g., a TCP socket).

At the beginning of a session, a client $c$ connects to a partition $p$ in one DC according to some load balancing scheme. This DC is referred to as the local DC. The partition $p$ serves all $c$'s operations. If $p$ does not store a key $k$ targeted by an operation, $p$ transparently forwards the operation to a replica of $k$. $c$ does not issue the next operation until it receives the reply to the current one.

\noindent{\bf Availability.} We use the term availability to indicate the ability of a system to never block a client operation in a presence of a network partition among DCs~\cite{Brewer:2000}.

\subsection{APIs}
\ts{}'s programming interface offers the following operations for interactive read-write transactions:

\pvs
~\\\noindent{$\bm {\bullet <T_{ID}, S> \gets START-TX():}$} starts an interactive transaction T and returns T's transaction identifier T{$_{ID}$} and the causal snapshot S visible to T. 

\pvs
~\\\noindent{$\bm {\bullet \langle vals\rangle \gets READ(T_{ID}, k_1, . . . , k_n):}$} reads in parallel the set of items corresponding to the input set of keys for a transaction identified by T{$_{ID}$}. 

\pvs
~\\\noindent{ $\bm {\bullet WRITE(T_{ID}, \langle k_1,v_1 \rangle,...,\langle k_n,v_n\rangle):}$} updates a set of keys, given as input, to the corresponding values for a transaction with T{$_{ID}$}.

\pvs
~\\\noindent{$\bm {\bullet COMMIT-TX(T_{ID}) :}$} finalizes a transaction with T{$_{ID}$} to atomically update items that have been modified by means of a WRITE operation in the scope of the transaction, if any. 

\pvs
~\\Upon a start of  a transaction $T$, clients can issue multiple read and write operations that can operate on multiple keys, before committing $T$.

\pvs
~\\Since, under the TCC programming model, conflicting updates are resolved rather than forbidden, in \ts{} transactions never abort due to conflicts. Although transactions can abort by means of system-related issues, e.g., not enough space on a server to perform an update, we do not consider aborts in this paper for simplicity reasons.

%% file: sections/design.tex
\section{Design of \ts{}}
\label{sec:design}

 The main goal of \ts{} is to implement TCC in a partially replicated sharded system, while providing non-blocking parallel read operations and achieving scalability. 
 We show the challenges to achieving these goals in Section~\ref{sec:design:challenges}. Next, in Section~\ref{sec:design:nonblocking} we present how \ts{} overcomes these challenges by a novel dependency tracking protocol and the use of a small client-side cache. Finally, in Section~\ref{sec:design:fault_tol} we discuss fault tolerance and availability in \ts{}. 

\subsection{Challenges of partial replication}
\label{sec:design:challenges}

Since TCC must simultaneously guarantee the preservation of causal consistency and the atomicity of multi-item writes, achieving non-blocking reads while maintaining TCC is challenging. In a fully replicated environment, the task of enforcing this behavior, is eased by two invariants: $i)$ all remote updates are received by all DCs, and hence every DC receives the dependencies of each update, and $ii)$ all updates of a transaction are performed within the same DC, and hence all the updates of the same transaction can be found in each DC. 
As example, assume that $X \leadsto Y$, and that $X$ and $Y$ are the latest versions of their keys, $x$ and $y$, respectively. Causal consistency dictates that if a client reads $x$ and $y$ in a transaction, then if $Y$ is returned, $X$ has to be returned as well. Furthermore, assume that $Z$, the last version of key $z$, has been written by the same transaction as $Y$. TCC further implies that either both $Y$ and $Z$ are visible to the transaction, or none of them is. 
Hence, some sort of communication among partitions in the {\em same} DC  is enough to ensure that $Y$ is visible in the DC only after $X$, and that $Y$ and $Z$ are atomically visible. 

Partial replication, instead, violates the two invariants described above. This leads to a new set of challenges of enforcing consistency and atomicity. 
First, tracking consistency is harder. In the previous example with keys $x$ and $y$, $X$ may be replicated from DC$_0$ to DC$_1$, and $Y$ from DC$_0$ to DC$_2$. Then, assume that a transaction from DC$_3$ reads $x$ in DC1 and $y$ in DC$_2$. The transaction has to ensure that $Y$ is read in DC$_2$ only if also $X$ is read, concurrently, in DC$_1$. 
Similarly, enforcing atomicity is harder. Assume that $Z$ is replicated from DC$_0$ to DC$_1$, and $Y$ from DC$_0$ to DC$_2$, and that a transaction from DC$_3$ reads $y$ and $z$. Then, the transaction in DC$_3$ has to ensure that either both $Y$ and $Z$ or none of them are read in an atomic fashion.

Addressing these two challenges is made more difficult by the fact that a read operation can target any replica of the target key. Hence, consistency and atomicity have to be preserved despite the fact that different transactions targeting the same keys can hit different replicas of those keys. The complexity of the problem is further exacerbated by the fact that different replicas of a version $X$ may be in different DCs that store different sub-sets of the dependencies of $X$.

One possible solution to these challenges could be allowing more than one round of client-server communication to perform a single parallel read operation. Servers can return possibly inconsistent versions in the first round(s), and the client can detect and fix these violations by issuing additional read requests~\cite{Lloyd:2011,Lloyd:2013,Mehdi:2017}.

Another possible solution could be blocking a read on a partition until the partition knows that all other involved partitions are serving the read operations from the same causal snapshot of the data store~\cite{Almeida:2013,Du:2014,Akkoorath:2016}. Clearly, these solutions increase the latency experienced by the transactions, and reduce the achievable throughput, because they introduce waiting times or require extra communication.

\subsection{Non-blocking reads in partial replication by \ts{}}
\label{sec:design:nonblocking}
\ts{} addresses these challenges by a combination of a novel dependency tracking protocol, called UST, and a client-side cache. UST identifies snapshots of the data store that can be read by  transactions without blocking. These snapshots are such that they have been already installed by every DC, so they are slightly in the past. The client-side cache stores the versions written by the client that are not yet reflected in the snapshot determined by UST. This allows clients to observe monotonically increasing snapshots even if UST identifies slightly stale snapshots. We now explain how \ts{} leverages UST and the client-side cache in its transactional protocol. 

\noindent{\bf Transactions in \ts{}.} \ts{} identifies key versions and snapshots by means of a scalar timestamp. Upon starting, a transaction is assigned a snapshot timestamp $st$ that, together with the content of the client-side cache, determines the snapshot visible to the transaction. Upon completing, each transaction that writes at least one key is assigned a commit timestamp that reflects causality, determined by means of a two-phase commit (2PC) protocol. 

\noindent{\bf Non-blocking reads in \ts{}.} The key idea in \ts{} is to identify a snapshot that has been installed by each DC. We define such snapshot {\em stable}.
A stable snapshot with timestamp $ts$ contains versions with a timestamp $\leq ts$, and indicates that {\em every} transaction $T$ with a commit timestamp $\leq ts$ has been applied in {\em every} data center that stores a replica of the key written by $T$.   

 Hence, a transaction can read from a stable snapshot without blocking or running multiple client-server rounds, regardless of the DC in which the individual reads of the transaction are performed.

 A coordinator partition is responsible to assign a stable snapshot to a transaction that starts.  Any node can act as the coordinator of any transaction. The coordinator enforces that the snapshots assigned to transactions issued by the same client advance monotonically. To this end, the client piggybacks its  last observed snapshot timestamp on the transaction start message.  

\noindent{\bf UST.} UST is the new protocol implemented by \ts{} to identify, in a scalable fashion, stable snapshots.  Each partition maintains a version vector that indicates the timestamps of the latest applied transactions, both the ones executed by the partition itself and the ones received from remote replicas. 
Periodically, partitions within the same DC and across DCs, by means of a gossiping protocol,  exchange the minimum of the timestamps in their version vectors. The aggregate minimum of the exchanged values identifies a timestamp such that all transactions with lower timestamps have been applied by the corresponding partitions in every DC. Namely, such aggregate minimum timestamp identifies the stable snapshot that transactions are assigned upon starting.

UST identifies stable causally consistent snapshots with a single timestamp, regardless of the scale of the system. This enables high scalability and efficiency, by reducing partition-to-partition and client-to-partition communication overhead.

\noindent{\bf Cache.} UST alone cannot enforce causality. In fact, the commit timestamp assigned to a transaction $T$ issued by client $c$ is higher than the stable snapshot assigned to $T$. On the one hand, this allows commit timestamps to reflect causality. On the other hand, it means that the commit timestamp of $T$ may be higher than the snapshot assigned by the next transaction issued by $c$. In that case, such snapshot would not include the modifications performed by $c$ in $T$, which may lead to violation of the read-your-own-write property required by causal consistency.

\ts{} overcomes this issue by storing on the client the versions written by the client. Upon receiving a snapshot timestamp $st$, a client $c$ removes from the cache all versions with timestamp $\leq st$. These versions are, in fact, included in the snapshots visible to any future transaction issued by $c$. Upon reading key $x$, $c$ first checks its cache. If a version exists in the cache, that version has to be read by $c$ to enforce the read-your-own-write property. Else, $c$ issues a read request to a replica.  In both cases, the read completes without blocking.

\noindent{\bf Generating timestamps.} As in recent proposals~\cite{Spirovska:2018,Mehdi:2017,Roohitavaf:2017,Gunawardhana:2017}, \ts{} uses Hybrid Logical Physical Clocks (HLC)~\cite{Kulkarni:2014} to generate timestamps. An HLC is a logical clock whose value on a partition is the maximum between the local physical clock and the highest timestamp seen by the partition plus one. Like logical clocks, HLCs can be moved forward to match the timestamp of an incoming event, without  blocking to wait that the local physical clock catches up with the timestamp of the event. Like physical clocks, HLCs advance in the absence of events and at approximately the same pace. Hence, HLCs improve the freshness of the snapshot determined by UST over a solution that uses logical clocks, which can advance at very different rates on different partitions.

\subsection{Fault tolerance}
\label{sec:design:fault_tol}

\noindent{\bf Failures (within a DC).} \ts{} can tolerate failures of a server by integrating existing solutions for 2PC-based systems, e.g., based on Paxos~\cite{Lamport:1998}. Reads are non-blocking also with such mechanisms enabled, because they access a snapshot corresponding to transactions that have been {\em already} committed. 

As in previous systems based on dependency aggregation protocols, the failure of a server blocks the progress of UST, but only as long as a backup has not taken over. 

 Client failures are transparent to the system. The clients only keep local meta-data, and cache data that have already been committed to the data-store. The contexts corresponding to transactions of failed clients are cleaned in the background after a timeout. 

\noindent{\bf Availability (among DCs).} \ts{} achieves availability in a DC as long as one replica per partition is reachable by a DC. In fact, remote operations can be performed by any DC, because the snapshot visible to a transaction is the same, regardless of the partition contacted to serve an operation. In addition, local operations never block. 

If all replicas of one partition cannot be reached by a DC, then \ts{} cannot complete remote operations that target that partition, thus leading to unavailability.

If a DC partitions from the rest of the system, then the UST freezes at all DCs, because it is computed as a system-wide minimum. As a result, transactions see increasingly stale snapshots of the data, and the client cache cannot be pruned.

%% file: sections/protocols.tex
\section{Protocols of \ts{}}
\label{sec:protocols}
We now describe the meta-data stored and the protocols implemented by the clients and servers in \ts{}.

\subsection{Meta-data}
\label{sec:protocols:meta}

\noindent{\bf Items.}
An item $d$ is represented as the tuple $ \langle k, v, ut, id_T, sr \rangle $. $k$ and $v$ are the key and value of $d$, respectively. $ut$ is the timestamp of $d$ which is assigned upon creation of $d$ and determines the snapshot to which $d$ belongs. $id_T$ is the id of the transaction that created the item version. $sr$ is the id of the DC where the item is created.

\pvs 
~\\\noindent{\bf Clients.}
In a client session, a client $c$ maintains the highest stable snapshot timestamp known by $c$, noted $ust_c$, and the commit time of its last update transaction, noted $hwt_c$. The client also maintains a private cache $WC_c$, which stores items written by $c$ that are not included in the stable snapshot yet. Finally, the client maintains the meta-data and data of the transaction that is currently running: $id_c$, which is the unique identifier of the transaction, and $WS_c$ and $RS_c$, which correspond to the transaction's write set and read set, respectively.
 
\pvs 
~\\\noindent{\bf Server.} A server $p_n^m$ is identified by the partition id ($n$), and the DC id ($m$), which is the local DC of the server. Additionally, $p_n^m$ also stores the replica id ($r$), where $r \leq R$, the replication factor of partition $p_n^m$.  

Each server has access to a monotonically increasing physical clock, $Clock_n^m$. The local clock value on $p_n^m$ is represented by the hybrid logical clock $HLC_n^m$. $p_n^m$ also maintains two vector clocks $VV_n^m$ and $GSV_n^m$, that represent vectors of HLCs. $VV_n^m$, has $R$ entries, one for each replica of partition $n$. $VV_n^m[i], i \neq r$, indicates the timestamp of the latest update received by $p_n^m$ that comes from the $i$-th replica of partition $n$. $VV_n^m[r]$ is the version clock of the server and represents the local snapshot installed by $p_n^m$. 
$GSV_n^m$, or Global Stabilization Vector, has $M$ entries. $GSV_n^m[i] = t$ means that $p_n^m$ is aware that all the nodes in the $m-$th data center have installed all events generated in the $i-$th data center with timestamp up to $t$.
The server also stores the UST of $p_n^m$, noted $ust_n^m$. $ust_n^m = t$ indicates that $p_n^m$ is aware that every partition in every DC has installed a snapshot with timestamp at least t.

Finally, as a standard practice for systems that perform a 2PC protocol, $p_n^m$ keeps two queues with prepared and committed transactions. The former, noted $Prepared_n^m$, stores transactions for which $p_n^m$ has proposed a commit timestamp and for which $p_n^m$ is waiting the commit message. The latter, noted $Committed_n^m$ stores transactions that have been assigned a commit timestamp and whose modifications are going to be applied to $p_n^m$.

\begin{algorithm}[t]
\algsize
\caption{Client $c$ (open session towards $p_n^m$).}
\label{alg:client}
\begin{algorithmic}[1]

\Function{Start}{}
\State send \param{StartTxReq \mib{ust_c}} to $p_n^m$
\State receive \param{StartTxResp \mib{id, ust}} from $p_n^m$
\State $id_c \lam id;\ ust_c \lam ust;$
\State $RS_c \lam \emptyset; WS_c \lam \emptyset$
\State Remove from $WC_c$ all items with commit timestamp up to $ust_c$

\EndFunction

\Statex

\Function{Read}{$\chi$}
\State $D \lam \emptyset;\ \chi' \lam \emptyset$
\ForEach{$k \in \chi$}
\State $d \lam$ check $WS_c,\ RS_c,\ WC_c$ (in this order)
\State {\bf if} $(d \neq NULL)\ {\bf then}\ D\gets d$ 
\EndFor
\State $ \chi' \lam \chi \setminus D$
\State send \param{ReadReq \mib{id_c, \chi'}} to $p_n^m$
\State receive \param{ReadResp \mib{D'}} from $p_n^m$
   \State $D \lam D \cup D'$
\State $RS_c \lam RS_c \cup D$
\State \Return $D$
\EndFunction

\Statex

\Function{Write}{$\chi$}
\ForEach{$\langle k,v \rangle \in \chi$}\Comment{Update $WS_c$ or write new entry}
\State {\bf if} $(\exists d \in WS: d == k) ${\bf then} $d.v \lam v$  {\bf else} $WS_c\lam WS_c \cup \langle k,v \rangle $
\EndFor
\EndFunction

\Statex

\Function{Commit}{}\Comment{Only invoked if $WS \neq \emptyset$}

\State send \param{CommitReq \mib{id_c, hwt_c, WS_c}} to  $p_n^m$
\State receive \param{CommitResp \mib{ct}} from $p_n^m$
\State $hwt_c \lam ct$\Comment{Update client's highest write time}

\State Tag $WS_c$ entries with $hwt_c$ 
\State Move $WS_c$ entries to $WC_c$ \Comment{Overwrite (older) duplicate entries}

\EndFunction

\end{algorithmic}
\end{algorithm}


\begin{algorithm}[t]
\algsize
\caption{Server $p_n^m$  - transaction coordinator.}
\label{alg:srv:coordinator}
\begin{algorithmic}[1]
\Event{receive \param{StartTxReq \mib{ust_c}} from $c$}
\State $ust_n^m\lam max\{ust_n^m, ust_c\}$\Comment{Update universal stable time}
\State $id_T \lam generateUniqueId()$
\State $TX[id_T] \lam ust_n^m $\Comment{Save TX context}
\State send \param{StartTxResp \mib{id_T, TX[id_T]}}\Comment{Assign transaction snapshot}
\EndEvent

\Statex
\Event{receive \param{ReadReq \mib{id_T, \chi}} from $c$}
		\State $ ust \lam TX[id_T]$
		\State $D \lam \emptyset$
		\State $\chi_i \lam \{k \in \chi : partition(k) == i\}$\Comment{Partitions with $\geq$ 1 key to read} 
		\For {($i: \chi_i \neq \emptyset$)} \Comment{Done in parallel}
        \State $j = getTargetDCForPartition(i)$ \Comment{Returns an id of a DC that replicates partition $i$}
		
		\State  send \param{ReadSliceReq \mib{\chi_i, ust}} to $p_i^j$
		    \State  receive \param{ReadSliceResp \mib{D_i}} from $p_i^j$
		
		\State $D \lam D \cup D_i$
		\EndFor		
	\State send \param{ReadResp \mib{D}} to $c$
\EndEvent

\Statex

\Event{receive \param{CommitReq\ \mib{ id_T, hwt, WS}} from $c$}
\State $\langle ust \rangle \lam TX[id_T]$
\State $ht \lam max\{ust, hwt\}$\Comment{Max timestamp seen by the client}
\State $D_i \lam \{\langle k, v \rangle \in  WS : partition(k) == i\}$ 
\For {($i: D_i \neq \emptyset$)}\Comment{Done in parallel}
    \State $j = getTargetDCForPartition(i)$ \Comment{Returns an id of a DC that replicates partition $i$}
    \State  send \param{PrepareReq \mib{id_T, ust, ht, D_i}} to $p_i^j$
    \State  receive \param{PrepareResp \mib{id_T, pt_i}} from $p_i^j$
\EndFor
\State ct $\lam max_{i : D_i \neq \emptyset}\{pt_i\}$\Comment{Max proposed timestamp}
\State {\bf for} {($i: D_i \neq \emptyset$)} {\bf do} send \param{CommitReq \mib{id_T, ct}} to $p_n^m$ {\bf end for}
\State delete TX[$id_T$]  \Comment{Clear transactional context of $c$}
\State send \param{CommitResp \mib{ct}} to $c$
\EndEvent

\end{algorithmic}
\end{algorithm}

\subsection{Operations}
\label{sec:protocols:ops}
Algorithm 1 reports the client protocol. Algorithm 2 and Algorithm 3 report the protocols executed by a server to run a transaction, for the cases in which the server is or is not the transaction coordinator, respectively. Algorithm 4 describes the replication  and the UST protocols.

\noindent{\bf Start.} Client $c$ starts a transaction $T$ by picking at random a coordinator partition (denoted $p_n^m$) in the local DC and sending it a start request with $ust_c$. $p_n^m$ uses $ust_c$ to update $ust_n^m$, so that $p_n^m$ can assign to the new transaction a snapshot that is at least as fresh as the one accessed by $c$ in previous transactions. $p_n^m$ uses its updated value of $ust_n^m$ as snapshot for $T$.  $p_n^m$ also generates a unique identifier for $T$, denoted $id_T$, and inserts $T$ in a private data structure. $p_n^m$ replies to $c$ with $id_T$ and the snapshot timestamp $ust_n^m$.

Upon receiving the reply, $c$ updates $ust_c$ and evicts from the cache any version with timestamp lower than    or equal to $ust_c$. $c$ can prune  such versions because the UST protocol ensures that they are included in the snapshot installed by any partition in the system. This means that if, after pruning, there is a version X in the private cache of $c$, then $X.ct > ust$ and hence the freshest version of $x$ visible to $c$ is $X$.


\begin{algorithm}[t]
\algsize
\caption{Server $p_n^m$  - transaction cohort.}
\label{alg:srv:cohort}
\begin{algorithmic}[1]

\Event{receive \param{ReadSliceReq\ \mib{\chi, ust}} from $p_i^j$}\label{alg:slice}

\State $ust_n^m\lam  max\{ust_n^m, ust\}$\Comment{Update universal stable time}

\State $D \lam \emptyset$

\For{($k \in \chi$)}

\State $D_{sv} \lam \{ d: d.k == k \land d.ut \leq ust\}$\Comment{ Universally visible}
\State $D\lam D\cup\{argmax_{d.ut} \{ d \in D_{kv}\}\}$\Comment{Freshest visible vers. of $k$}
\EndFor
\State send \param{SliceResp \mib{D}} to $p_i^j$
\EndEvent

\Statex 

\Event{receive \param{PrepareReq\ \mib{id_T, ust, ht, D_i}}} from $p_i^j$
\State $HLC_n^m \lam max(Clock_n^m, ht+1,HLC_n^m+1)$\Comment{Update HLC}
\State $ust_n^m \lam max\{ust_n^m,ust\}$\Comment{Update universal stable time}
\State pt $\lam max\{HLC_n^m, ust_n^m\}$\Comment{Proposed commit time}
\State $Prepared_n^m \lam Prepared_n^m \cup \{ id_T, pt, D_i \}$\Comment{Append to pending list}
\State  send \param{PrepareResp \mib{id_T, pt}} to $p_i^j$
\EndEvent

\Statex

\Event{receive \param{CommitReq\ \mib{id_T, ct}}} from $p_i^j$
\State $HLC_n^m \lam max(HLC_n^m, ct, Clock_n^m)$\Comment{Update HLC}
\State $\langle id_T, pt, D \rangle \lam \{ \langle i, r,\phi  \rangle  \in Prepared_n^m :  i == id_T\}$
\State $Prepared_n^m \lam Prepared_n^m\setminus\{\langle id_T, pt, D\rangle\}$\Comment{Remove from pending}
\State $Committed_n^m \lam Committed_n^m \cup \{ \langle id_T, ct, D \}$\Comment{Mark to commit}

\EndEvent
\end{algorithmic}
\end{algorithm}

\begin{algorithm}[t!]
\algsize
\caption{ Server $p_n^m$  - Auxiliary functions.}
\label{alg:srv:aux}
\begin{algorithmic}[1]

\Function{update}{\mib{k,v,ut,id_T}}
\State  create $d: \langle d.k, d.v, d.ut, id_T, d.sr \rangle \lam \langle k, v, ut, id_T, m \rangle $
\State  insert new item $d$ in the version chain of key $k$
\EndFunction

\Statex

\Event{Every $\Delta_R$}
\State {\bf if}\ ($Prepared_n^m \neq \emptyset$)\ {\bf then}\ $ub \lam min_{\{p.pt\}}\{ p \in Prepared_n^m\} - 1$
\State {\bf else}\ $ub \lam max\{Clock_n^m, HLC_n^m\}$\ {\bf end if}

\State $\rho_n \lam Replicas(n)$ 

\If{($Committed_n^m \neq \emptyset$)}\Comment{Commit tx by increasing order of $ct$}
    \State $C \lam \{ \langle id, ct, D \rangle\} \in Committed_n^m : ct < ub$
    \For{($ T \lam \{\langle id, D \rangle\} \in (group\ C\ by\ ct$))}
    \For{($\langle id, D \rangle \in T$)}
        \State {\bf for}\ ($\langle k, v \rangle \in D$)\ {\bf do}\ update (\mib{k, v, ct, id})\ {\bf end for}
    \EndFor
    \State {\bf for}\ ($j \in \rho_n \land j \neq r $)\ {\bf do}\ send \param{Replicate \mib{T, ct}} to $p_n^j$\ {\bf end for}
    \State $Committed_n^m \lam Committed_n^m \setminus T$
\EndFor
    \State $VV_n^m[m] \lam ub$\Comment{Set version clock}
\Else
    \State $VV_n^m[m] \lam ub$\Comment{Set version clock}

    \State {\bf for}\ ($j \in \rho_n \land j \neq r $)\ {\bf do}\ send \param{Heartbeat \mib{VV_n^m[m]}} to $p_n^j$ \ {\bf end for} \Comment{Send heartbeat to replicas}

\EndIf
\EndEvent

\Statex
\Event{receive \param{Replicate \mib{T, ct}} from $p_n^j$}

\For{($\langle id, D \rangle \in T$)}
\For{($\langle k, v \rangle \in D$)}
\State update (\mib{k, v, ct, id})
\EndFor
\EndFor
\State $i \lam GetReplicaIdForDC(j)$
\State $VV_n^m[i] \lam ct$\Comment{Update version clock of i-th replica of n-th partition}

\EndEvent

\Statex
\Event{receive \param{Heartbeat t} from $p_n^j$}
\State $i \lam GetReplicaIdForDC(j)$
\State $VV_n^m[i] \lam t$\Comment{Update version clock of i-th replica in n-th partition}

\EndEvent

\Statex
\Event{every $\Delta_G$ time}\label{alg:gsv} \Comment{Gather global stable times from other DCs}
\State $GSV_n^m[j] \lam min\{VV_i^m[j]\}, \forall j=0\ldots M-1, \forall i = 0\ldots N-1$
\EndEvent

\Statex
\Event{every $\Delta_U$ time}\label{alg:ust} \Comment{Compute universal stable time}
\State $minGST \lam min\{GSV_i^m[j]\}, \forall j=0\ldots M-1, \forall i = 0\ldots N-1$\label{alg:usv:min}
\State $ust_n^m \lam max\{minGST, ust_n^m\}$\label{alg:usv:mono}\Comment{Enforce monotonicity of $ust_n^m$}
\EndEvent

\end{algorithmic}
\end{algorithm}

\pvs
~\\\noindent{\bf Read.}
For each key $k$ to read, $c$ searches the write-set, the read-set and the client cache, in this order. If an item corresponding to $k$ is found, it is added to the set of items to return, ensuring read-your-own-writes and repeatable-reads semantics. Reads for keys that cannot be served locally are sent to the transaction coordinator $p_n^m$ together with the transaction id. $p_n^m$ retrieves the snapshot corresponding to the transaction, and sends to each involved partition the set of keys to be read, in parallel. Because each DC only stores a subset of the full data set, some of the contacted partitions may belong to a remote DC that replicates the partitions where the keys belong. Remote DCs can be chosen depending on geographical proximity or on some load balancing scheme. 

Upon receiving a read request, regardless of whether it originates from the local DC or from a remote one, $p_n^m$ first updates its $ust_n^m$, if it is smaller than the transaction's snapshot (Alg. 3 Line 2). Next, the server returns, for each key, the version within the snapshot with the highest timestamp (Alg. 3 Lines 4--7).
As we shall see shortly, the commit protocol of \ts{} allows concurrent updates on the same key, both within a DC and in different DCs. This is typically the case in TCC systems to avoid costly validation protocols for update transactions~\cite{Akkoorath:2016,Spirovska:2018}. In case multiple versions of a key are assigned the same timestamp, \ts{} totally orders versions by a concatenation of timestamp, transaction id and source data center id, in this order. Once $p_n^m$ has received the reply from all the partitions contacted, $p_n^m$ sends the items to the client, which inserts them in its read-set.

\pvs 
~\\\noindent{\bf Write.}
Client $c$ locally buffers the writes in its write-set $WS_c$. If a key being written is already present in $WS_c$, then it is updated; otherwise, it is inserted in $WS_c$.

\pvs
~\\\noindent{\bf Commit.}
To finalize the transaction, the client sends a commit request to the coordinator with the content of $WS_c$, the id of the transaction and the commit timestamp of its last update transaction $hwt_c$, if any. The commit protocol of \ts{} is based on the two-phase commit (2PC) protocol. The coordinator contacts the partitions that store the keys that need to be updated  and sends them the corresponding updates and $hwt_c$. Note that some of the contacted partitions can belong to a remote DC. Each partition involved,  first updates its clock to be at least as high as the maximum between the transaction's snapshot timestamp and $hwt_c$. Then, each partition increases its clock and sends the updated clock value to the coordinator as a commit timestamp.  The proposed timestamp reflects causality because it is higher than both the snapshot timestamp and $hwt_c$.  Each partition also inserts the transaction id, the set of keys to be modified on the partition and the proposed timestamp in the queue of  pending transactions.  

The coordinator  picks the maximum among the proposed timestamps, sends it to the partitions involved in the transaction, clears the local context of the transaction and sends the commit timestamp to the client. Upon receiving the commit message, a partition increases its clock to match the commit time, if needed, and moves the transaction from the pending queue to the commit one, with the new commit timestamp.

\pvs
~\\\noindent{\bf Applying and replicating transactions.} 
Periodically, the effects of  transactions committed by  $p_n^m$ are applied on the $p_n^m$, in increasing commit timestamp order (Alg. 4 Lines 6-21). $p_n^m$  applies the modifications of transactions that have a commit timestamp strictly lower than the lowest timestamp present in the pending list. This timestamp represents the lower bound on the commit timestamps of future transactions on $p_n^m$. After applying the transactions, $p_n^m$ updates its local version clock and replicates the update operations in the applied transactions to its remote replicas.

If $p_n^m$ does not commit a transaction for a given amount of time, $p_n^m$ updates its local clock, and sends it to its peer replicas by means of a heartbeat message. This ensures the progress of the UST also in absence of updates.

\pvs
~\\\noindent{\bf Stabilization protocol.} Every $\Delta_G$ time units, partitions within a data center exchange the minimum of their version vectors to compute the global stable time ($GST$) of the local data center. Similarly to previous work~\cite{Du:2014,Akkoorath:2016}, \ts{} organizes nodes within a DC as a tree to reduce message exchange. The $GST$ is progressively aggregated from the leaves to the root, and then propagated from the root to all the nodes in the DC. Next, all the roots from each DC exchange their $GST$ values. 

Every $\Delta_U$ time units, the roots compute the $ust_n^m$ as the aggregate minimum of the received $GST$s  and propagate it to all the other nodes in the DC. 

\pvs

~\\\noindent{\bf Garbage collection.} 
Periodically the partitions in the system exchange the oldest snapshot corresponding to an active transaction ($p_n^m$ sends its current stable snapshot timestamp if it has no running transactions). The aggregate minimum determines the oldest snapshot $S_{old}$ that is visible to a running transaction. The partitions scan the version chain of each key backwards and keep the all the versions up to (and including) the oldest one within $S_{old}$. Previous versions are removed. The same protocol that computes the UST also computes $S_{old}$.

\setlength{\textfloatsep}{15pt}

\subsection{Correctness}
\label{sec:protocols:correctness}
We now provide an informal proof sketch that \ts{} provides causal consistency by showing that $i)$ reads observe a causally consistent snapshot and $ii)$ writes are atomic.

\begin{lemma}
\label{lemmas:1}
The snapshot time $sn_T$ of a transaction $T$ is always lower than the commit time of $T$, $sn_T < T.ct$.
\end{lemma}

\begin{proof}
Let $t$ be a transaction with snapshot time $sn_T$ and commit time $ct$. The snapshot time is determined during the start of the transaction (Alg. 2 Line 4). The commit time is calculated in the commit phase of the 2PC protocol, as maximum value of the proposed prepare times of all partitions participants in the transaction (Alg. 2 Line 26). In order to reflect causality when proposing a prepare timestamp, each partition proposes higher timestamp than the snapshot timestamp (Alg. 3 Line 12). Thus, the commit time of a transaction, $ct$, is always greater than the snapshot time, $sn_T$.
\end{proof}

\begin{figure*}[b!]
  \begin{subfigure}[h]{0.5\textwidth}
        \centering
       \includegraphics[scale = 0.4]{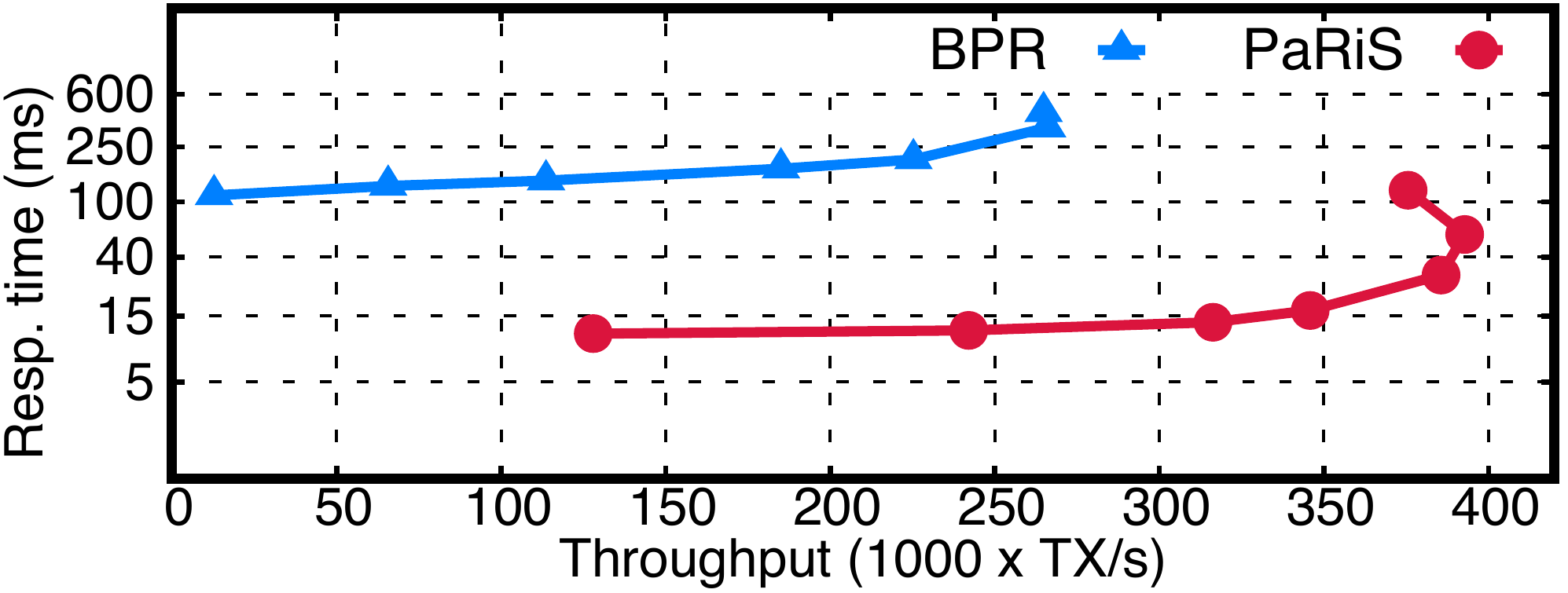}
        \caption{Throughput vs average TX latency (95:5 r:w ratio).}
        \label{fig:xput_lat:read_intensive}
    \end{subfigure}
    \begin{subfigure}[h]{0.5\textwidth}
    \centering
       \includegraphics[scale=0.4]{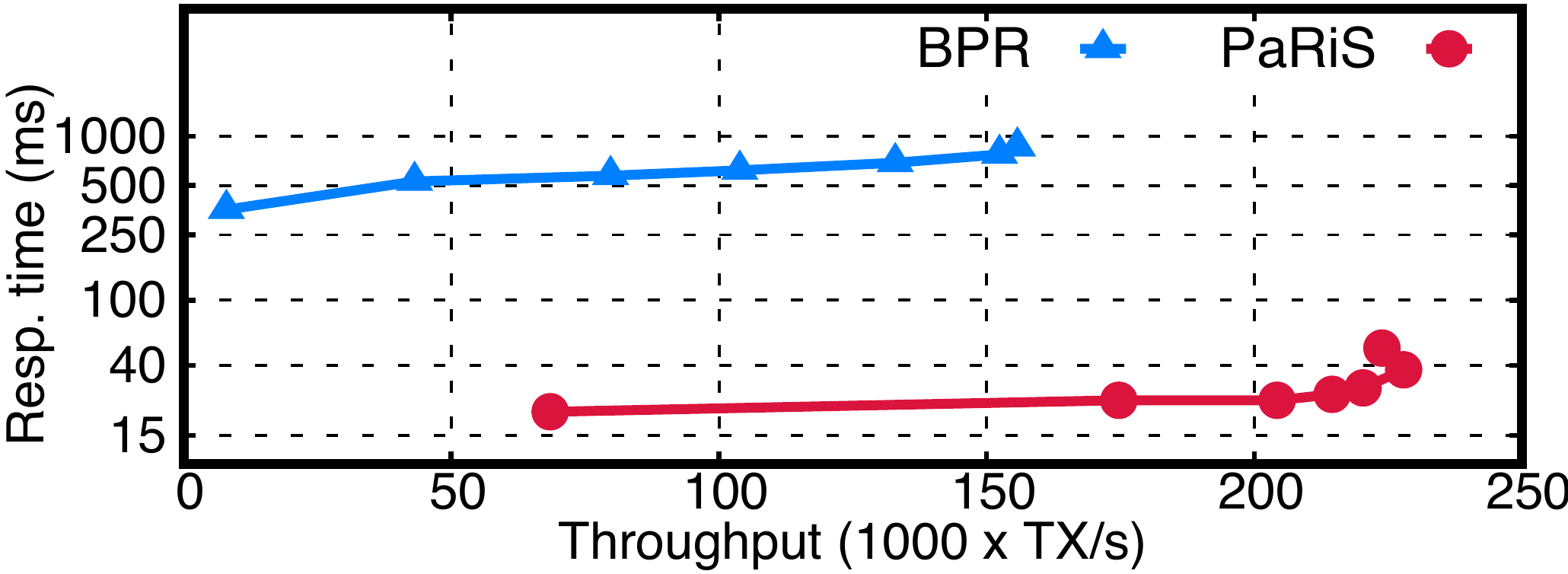}
        \caption{Throughput vs average TX latency (50:50 r:w ratio).}
        \label{fig:xput_lat:write_intensive}
    \end{subfigure}
\caption{Performance of \ts{} and BPR (logarithmic scale) with different 95:5 (a) and 50:50 (b) r:w ratios, 4 partitions involved per transaction (5 DCs, 45 partitions, replication factor is 2, 18 machines per DC). \ts{} outperforms BPR for  both read-heavy and write-heavy workloads.
}\label{fig:xput_lat}
\end{figure*}

\begin{proposition}
\label{propositions:1}
If an update $u_2$ causally depends on an update $u_1$, $u_1\leadsto u_2$, then $u_1.ut < u_2.ut$.
\end{proposition}

\begin{proof}
Let $c$ be the client that wrote $u_2$. There are three cases upon which $u_2$ can depend on $u_1$, described in Section~\ref{sec:model:cc}: \textit{1}) $c$ committed $u_1$ in a previous transaction; \textit{2}) $c$ has read $u_1$, written in a previous transaction and \textit{3}) $c$ has read $u_3$, and there exists a chain of direct dependencies that lead from $u_1$ to $u_3$, i.e. $u_1 \leadsto ... \leadsto u_3$ and $u_3 \leadsto u_2$. 

\textit{Case 1}. When a client commits a transaction, it piggybacks the last update transaction commit time $hwt_c$, if any, to its commit request for the transaction coordinator (Alg. 1 Line 27) which is, furthermore, piggybacked as $ht$ in its prepare requests to the involved partitions (Alg. 2 Line 23).
To reflect causality when proposing a commit timestamp, each partition proposes higher timestamp than both $ht$ and the snapshot timestamp (Alg. 3 Lines 10--14). The coordinator of the transaction chooses the maximum value from all proposed times from the participating partitions (Alg. 2 Line 26) to serve as commit time, $ct$, for all the updated items in the transaction. The new version of the data item is written in the key-value store with $ct$ as its update time, $ut$ (Alg. 4 Lines 2 and 13).
When $c$ commits the transaction that updates $u_2$, it piggybacks the commit time of the transaction that updated $u_1$. Hence, from the discussion above it follows that $u_1.ct < u_2.ct$. Because the commit time of a transaction is the update time of all the data item versions updated in the transaction (Alg. 4 Lines 2 and 13), we have $u_1.ut < u_2.ut$.

\textit{Case 2:} $c$ could have read $u_1$ either from $c$'s client cache or from the transaction causal snapshot $sn_T$. 

If $c$ has read $u_1$ from  $c$'s client cache, then $c$ has written $u_1$ either in a previous transaction in the same thread of execution or in the current one. If $c$ wrote $u_1$ in the same transaction where $u_2$ is also written, then it is not possible to have $u_1 \leadsto u_2$ because all the updates from that transaction will be given the same commit, i.e. update timestamp, indicating that $u_1.ut = u_2.ut$. Thus, $u_1$ must be written in a previous transaction and from \textit{Case 1} it follows that $u_1.ut < u_2.ut$.

Next, we will consider the case when $c$ read $u_1$ from the causal snapshot $sn_{T}$ that contains $u_1$.
When a transaction $T$ is started, the snapshot $sn_T$ is determined by Alg. 2 Line 2, $sn_T = max\{ust_c, ust_n^m\}$. From Alg. 3 Line 5 we have that $u_1.ut \leq ust = sn_T$. From Lemma~\ref{lemmas:1} it follows that $u_1.ut \leq sn_t < u_2.ct= u_2.ut$. Therefore, $u_1.ut < u_2.ut$.

\textit{Case 3:} If $u_2$ depends on $u_1$ because of a transitive dependency out of $c's$ thread-of-execution, it means that there exists a chain of direct dependencies that lead from $u_1$ to $u_2$, i.e., $u_1 \leadsto ... \leadsto u_3$ and $u_3 \leadsto u_2$. Each pair in the transitive-chain, belongs to either  \textit{Case 1} or \textit{Case 2}. Hence, the proof of \textit{Case 3} comes down to chained application of the correctness arguments from \textit{Case 1} and \textit{Case 2}, proving that each element has an update time lower than its successor's.
\end{proof}

\begin{proposition}
\label{propositions:2}
A partition vector clock $VV_n^m[i] = t$ implies that $p_n^m$ has received all updates from $i-th$ replica with commit time, $ct \leq t$.
\end{proposition}
\begin{proof}
We need to show that this proposition is valid for both local and remote updates. To prove the former, we show that there are no pending local updates with commit timestamp $ct \leq t$. When $p_n^m$ updates the local replica vector clock entry $VV_n^m[r]$, it finds the minimum prepare timestamp of all transactions that are currently in the prepare phase (Alg 4. Line 6). Because the commit time is calculated as the maximum of all prepare times (Alg 2. Line 26) and the HLC clock is monotonic (Alg. 3 Lines 10 and 16), it is guaranteed that all future transactions will receive a commit time which is greater than or equal to this minimum prepare timestamp. So, when the $VV_n^m[r]$ is set to the minimum of the prepare times of all transactions in the prepare queue minus 1 (Alg. 4 Line 6), $p_n^m$ has already received all updates for the snapshot $VV_n^m[i] = t$.

To show the latter, we use prove by contradiction. Assume there is a remote update $u$ from $i-th$ replica such that $u.ct < t$, and $p_n^m$ has not received $u$. By Alg. 4 Line 30, the partition would have received an update $u'$ such that $u'.ct = t$. The updates are sent in the order of their commit timestamps (Alg. 4 Lines 9--16). Hence, if $u'.ct > u.ct$ the $p_n^m$ could not have received another update $u'$ before $u$. Therefore, $u.ct > t$, implying that $u.ct \nless t$, leading to the contradiction.
\end{proof}

\begin{proposition}
\label{propositions:3}
Snapshots in \ts{} are causal. 
\end{proposition}
\begin{proof} 
To start a transaction, a client $c$ piggybacks the freshest snapshot it has seen, ensuring the monotonicity of the snapshot seen by $c$ (Alg. 2 Line 2). Commit timestamps reflect causality (Alg. 2 Line 26), and UST tracks a lower bound on the snapshot installed by every partition in all DCs (Alg. 4 Lines 36-38). If $X$ is within the snapshot of a transaction, so are its dependencies (Proposition~\ref{propositions:1}). On top of the snapshot provided by the coordinator, $c$ also can read the writes, that are not yet included in the snapshot, from the cache. These writes cannot depend on items created by other clients that are outside the snapshot visible to $c$.
\end{proof}

\begin{proposition}
\label{propositions:4}
Writes are atomic.
\end{proposition}
\begin{proof}
Although updates are made visible independently on each partition $p_n^m$ involved in the commit phase, either all updates are made visible or none of them are, i.e. the atomicity is not violated. All updates from a transaction belong to the same snapshot because they all receive the same commit timestamps (Alg. 2 Line 27). The updates are being installed in the order of their commit timestamps (Alg. 4 Lines 9--16). The visibility of the item versions is determined by the transaction snapshot (Alg. 3 Line 5), which is based on the value of $p_n^m$'s universal stable time $ust_n^m$. $ust_n^m$ is computed by the UST protocol as the aggregate minimum of the version vectors entries of all partitions of all data centers (Alg. 4 Lines 34--38).
\end{proof}

\ts{} implements TCC, as every transaction reads from a causally consistent snapshot (Proposition~\ref{propositions:3}) that includes all effects (Proposition~\ref{propositions:4}) of its causally dependent transactions.

%% file: sections/evaluation.tex
\section{Evaluation}
\label{sec:eval}
\begin{figure*}[t!]
  \begin{subfigure}[h]{0.5\textwidth}
        \centering
       \includegraphics[scale = 0.35]{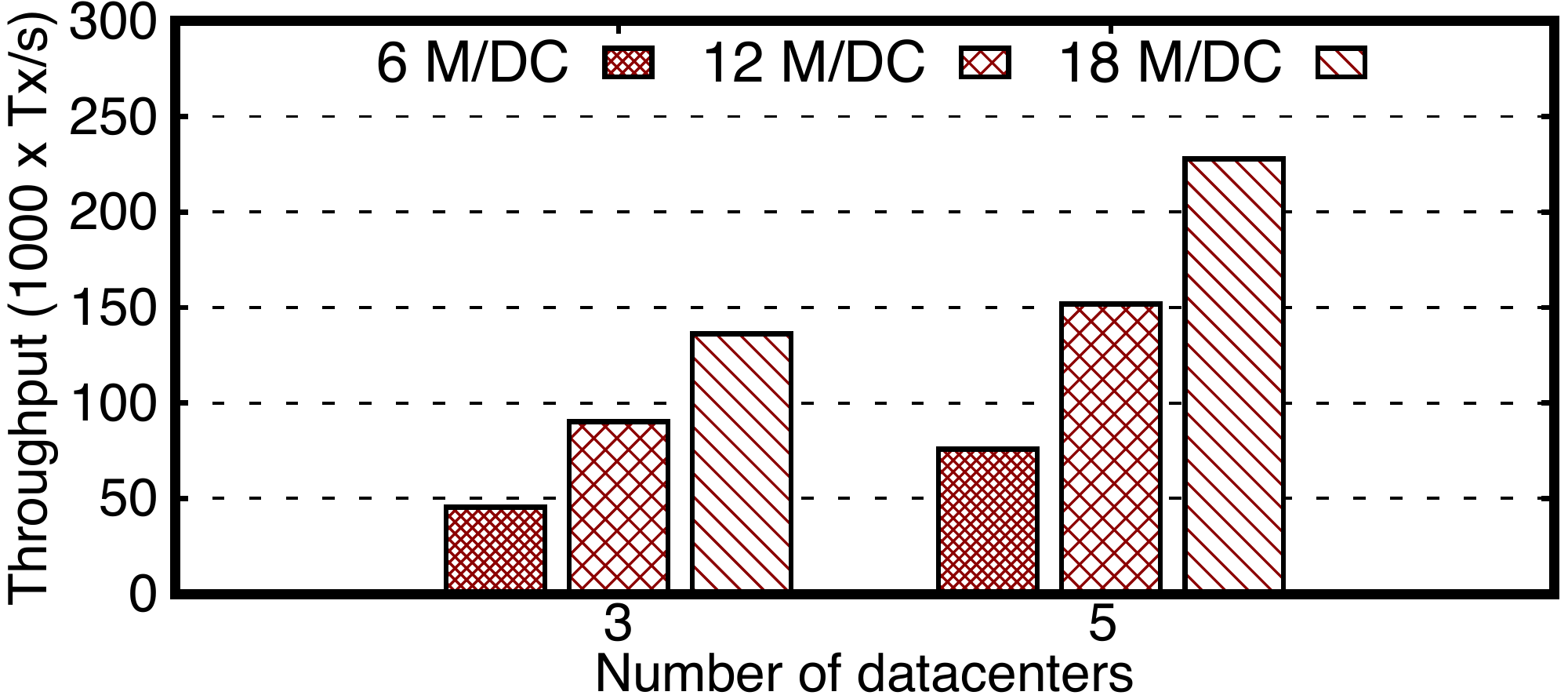}
        \caption{Throughput when varying the number of machines per DC.}
        \label{fig:scal:machines_dc}
    \end{subfigure}
    \begin{subfigure}[h]{0.5\textwidth}
    \centering
       \includegraphics[scale=0.35]{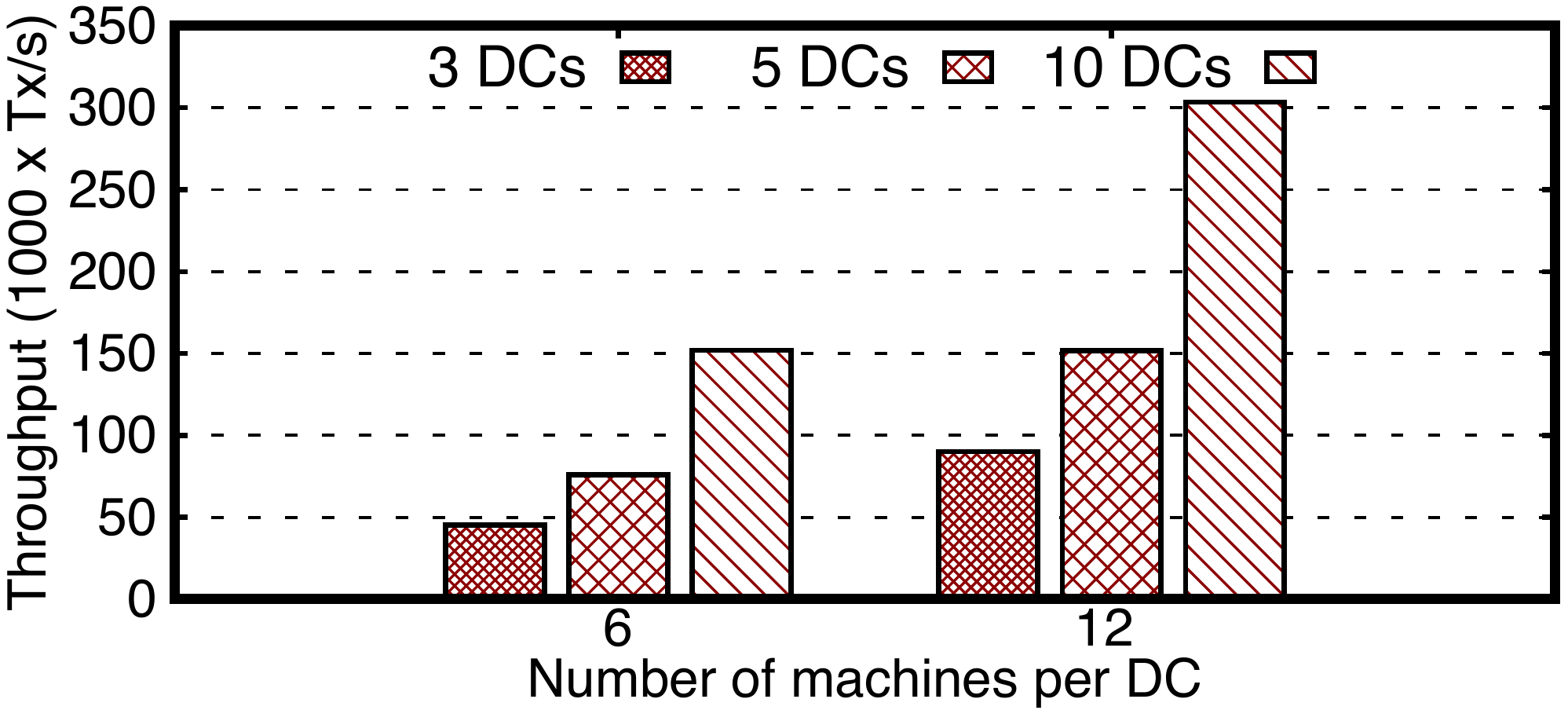}
        \caption{Throughput when varying the number of DCs.}
        \label{fig:scal:dcs}
    \end{subfigure}
\caption{Throughput achieved by \ts{} when increasing the number of machines per DC (a) and DCs (b). \ts{} achieves good scalability both when increasing the number of machines per DCs and DCs.
}\label{fig:scal}
\end{figure*}

\begin{figure*}[b!]
  \begin{subfigure}[h]{0.5\textwidth}
        \centering
       \includegraphics[scale = 0.35]{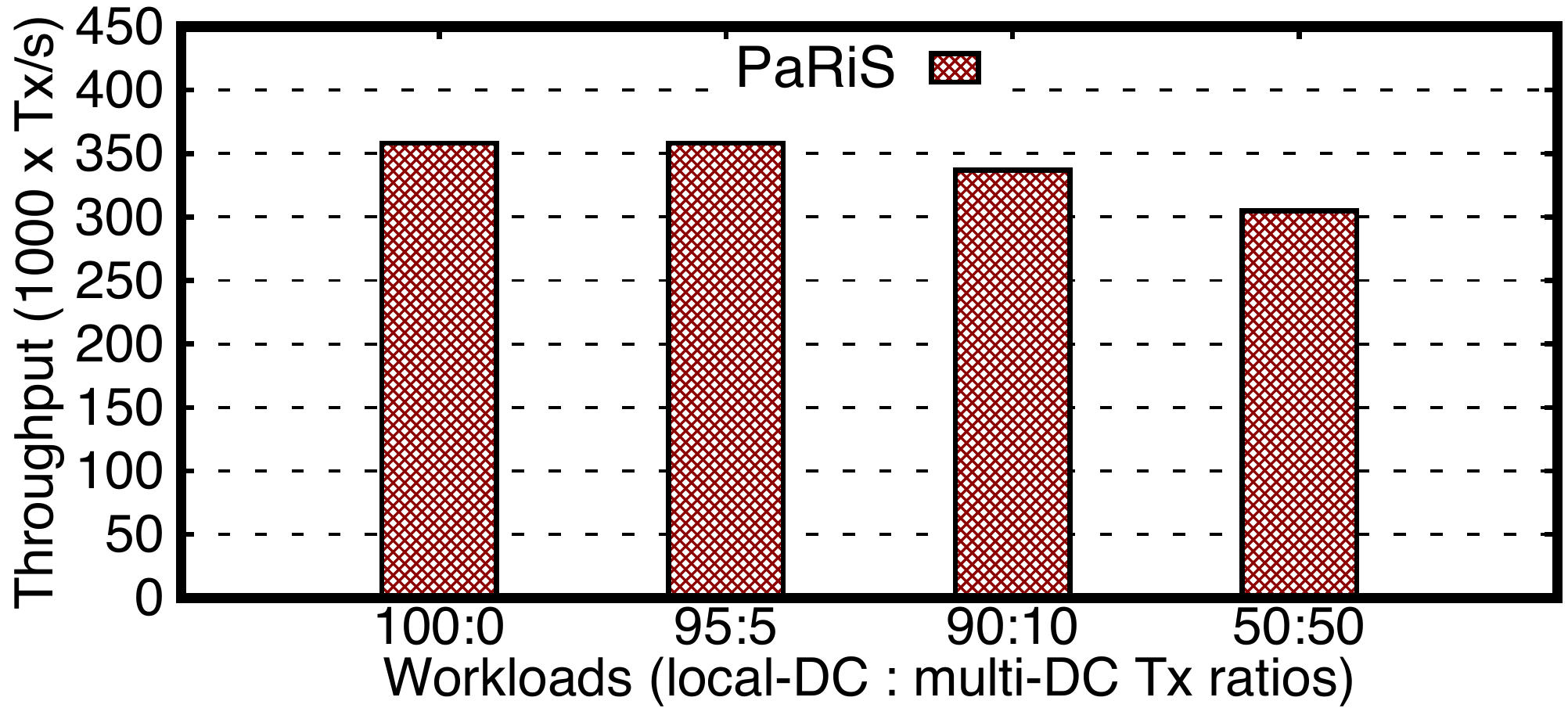}
        \caption{Throughput when varying the locality of the transactions.}
        \label{fig:locality:xput}
    \end{subfigure}
    \begin{subfigure}[h]{0.5\textwidth}
    \centering
       \includegraphics[scale=0.35]{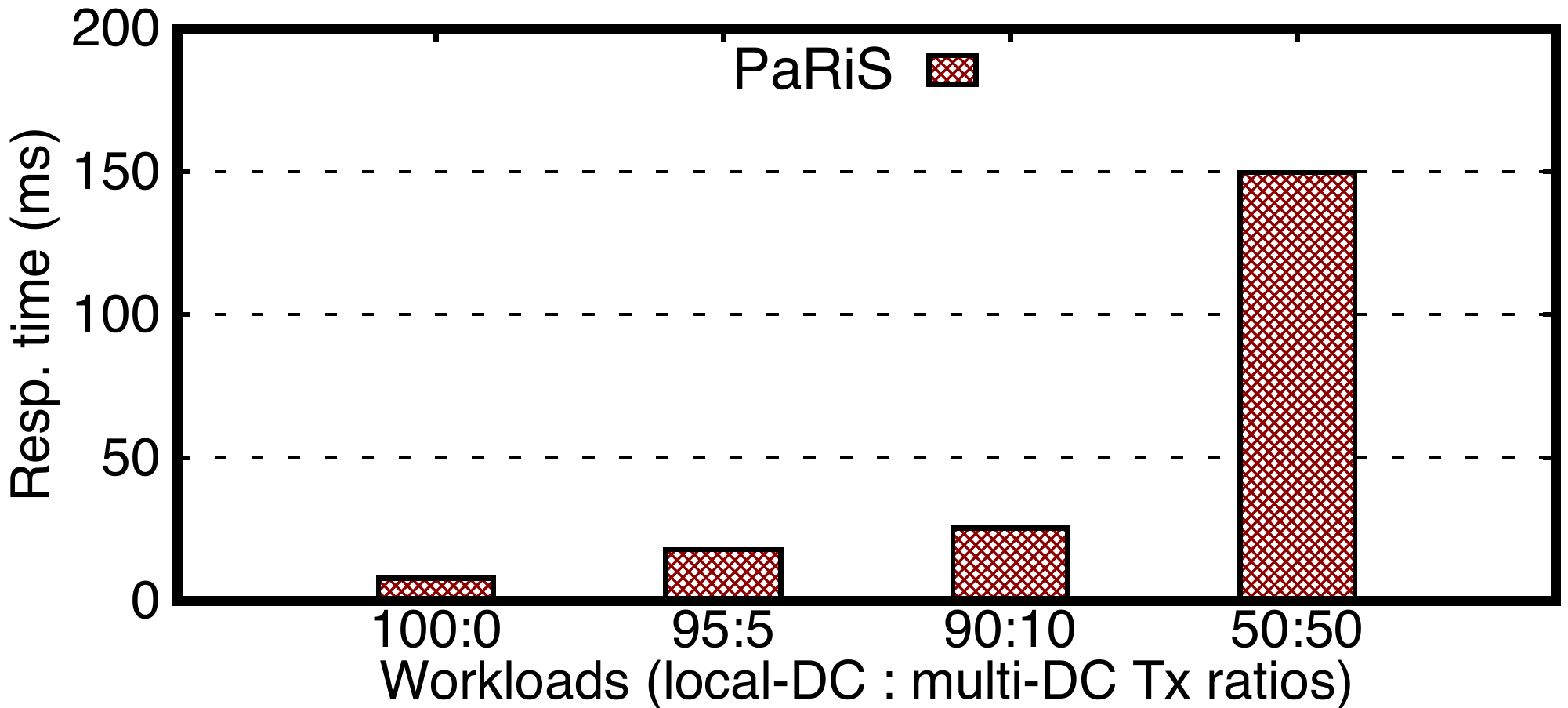}
        \caption{Latency when varying the locality of the transactions.}
        \label{fig:locality:lat}
    \end{subfigure}
\caption{Throughput (a) and latency (b) achieved by \ts{} when varying the locality of the transactions with 100:0, 95:5, 90:10 and 50:50 local-DC:multi-DC ratio.
}\label{fig:locality}
\end{figure*}

\noindent{\bf Competitor system.} To assess  the advantages of having non-blocking reads, we compare \ts{} against a blocking protocol, that we call Blocking Partial Replication, or BPR. In BPR the snapshot of a transaction $T$ of client $c$ is determined as the maximum of the highest causal snapshot seen by $c$ and the clock value of the transaction coordinator. 
BPR uses one timestamp to encode transactional snapshots, so we can compare fairly the resource efficiency of \ts{} versus the one of BPR. BPR also favors the freshness of the snapshots that are visible to transactions. BPR, however, implies having blocking transactional reads, because the server must ensure that the returned version  belongs to a causal snapshot. To this end, a read operation with snapshot timestamp $t$ is blocked on a partition until the partition has applied all local and remote transactions with timestamp up to $t$.

\subsection{Experimental environment}
\noindent{\bf Platform.} We consider a geo-replicated setting deployed across up to 10 replication sites on Amazon EC2 (North Virginia, Oregon, Ireland, Mumbai, Sydney, Canada, Seul, Frankfurt, Singapore and Ohio). When using 3 DCs, we use Virginia, Oregon and Ireland. When using 5 DCs, we use the previously mentioned 3 DCs plus Mumbai and Sydney. In each DC we use up to 18 servers (c5.xlarge instances with 4 VCPUs and 8 GB of RAM). The replication factor is 2. We choose this value because it allows us to use 3 as minimum number of DCs in our experiment and use partial replication. 

We spawn one client process per partition in each DC. Clients are collocated with the server partition they use as a transaction coordinator. The clients issue requests in a closed loop. To generate different load conditions, we spawn different number of threads per client process. Depending on the type of the workload or the protocol, a different number of threads is needed to saturate the target system.  
Each ``dot" in the curve plots we report corresponds to a different number of active threads per client process.

\pvs
~\\\noindent{\bf Implementation.} We implement \ts{} and BPR in the same C++ code-base. Both protocols implement the last-writer-wins rule for convergence. We use Google Protobufs for communication, and NTP to synchronize physical clocks. The stabilization protocols run every 5 milliseconds. 
\bigskip
\pvs
~\\\noindent{\bf Workloads.} We use workloads with 95:5 and 50:50 r:w ratios that correspond to the update-heavy (A) and read-heavy (B) YCSB workloads~\cite{Cooper:2010}. These are standard workloads also used to benchmark other TCC systems~\cite{Akkoorath:2016,Mehdi:2017,Zawirski:2015,Spirovska:2018}.  Transactions generate the three workloads by executing 19 reads and 1 write (95:5), and 10 reads and 10 writes (50:50).  Hence, in each workload a transaction executes 20 operations per transaction. A transaction first executes all the reads in parallel, and then all the writes in parallel. 

A transaction can target only partitions in the local DC, or can touch random partitions in remote DCs. In the first case, we say that a transaction is ``local-DC"; else, we say it is ``multi-DC". When accessing a remote partition, a client can choose among two replicas. We assign to every client in a DC the same preferred remote replica for each partition. We vary the preferred replica in the DCs using a round-robin assignment, to balance the load. To evaluate the effect of the partial replication, we use workloads with 100:0, 95:5, 90:10 and 50:50 local-DC:multi-DC ratios.

The default workload we consider uses the 95:5 r:w ratio, 95:5 local-DC:multi-DC ratio and runs transactions that involve 4 partitions on a platform deployed over 90 machines spread over 5 DCs. The default deployment has 45 partitions that are replicated with replication factor 2. Hence, each DC has a total of 18 machines. 

We also consider variations of this workload in which we change the value of one parameter and keep the others at their default values. Transactions access keys within a partition according to a zipfian distribution, with parameter 0.99, which is the default in YCSB and resembles the strong skew that characterizes many production systems~\cite{Atikoglu:2012,Nawab:2015,Balmau:2017}. We use small items (8 bytes), which are prevalent in many production workloads~\cite{Atikoglu:2012,Nawab:2015}.

\subsection{Latency and throughput}
\label{sec:eval:lat_xput}
\pvs
~\\\noindent{\bf Blocking vs. non-blocking.} Figure~\ref{fig:xput_lat:read_intensive} and Figure~\ref{fig:xput_lat:write_intensive} report the average transaction latency vs. throughput achieved by \ts{} and BPR with the 95:5 (the default) and with the 50:50 r:w ratios.
In the read-dominated case, \ts{} achieves up to 5.91x lower response times and up to 1.47x higher throughput than BPR. \ts{} also achieves up to 20.56x lower response times and up to 1.46x higher throughput than BPR in the write-dominated workload. 
\ts{} achieves lower latencies because it never has to wait for a snapshot to be installed. \ts{} achieves higher throughput because it does not incur any overhead to block/unblock read requests. Because BPR is a blocking protocol, it needs a higher number of concurrent client threads to fully utilize the processing power left idle by blocked reads. This creates more contention on the physical resources and more synchronization overhead to block and unblock reads, which ultimately leads to lower throughput. 

\pvs
~\\\noindent{\bf Blocking time.}
The average blocking time of the read phase of a transaction in BPR is 29 ms for the top throughput in the read-dominated workload (Figure~\ref{fig:xput_lat:read_intensive}) and 41 ms for the top throughput in the write-dominated workload (Figure~\ref{fig:xput_lat:write_intensive}).

\subsection{Scalability}
\label{sec:eval:scal}

\pvs
~\\\noindent{\bf Varying the number of machines per DCs.} Figure~\ref{fig:scal:machines_dc} reports the throughput achieved by \ts{} when using 6, 12 and 18 machines/DC. We consider two geo-replicated deployments that use 3 and 5 DCs. In both cases, \ts{} achieves the ideal improvement of 3x when scaling from 6 to 18 machines/DC. This result showcases the ability of \ts{} to scale horizontally regardless of the number of DCs on which it is deployed.

\pvs
~\\\noindent{\bf Varying the number DCs.} Figure~\ref{fig:scal:dcs} reports the throughput achieved by \ts{} when deployed on 3, 5 and 10 DCs. We consider two cases corresponding to 6 and 12 machines/DC. In both cases \ts{} achieves the ideal improvement of 3.33x, when scaling from 3 to 10 DCs. This result shows that \ts{} scales well to higher numbers of DCs for different sizes of the platform within each DC.

\subsection{Varying data access locality}
\label{sec:eval:wrkload}
Figure~\ref{fig:locality:xput} reports the maximum throughput achieved by \ts{} when varying the locality ratio (local-DC:multi-DC) of transactions from 100:0 to 50:50.  Figure~\ref{fig:locality:lat} shows the average transaction latency corresponding to the throughput values reported in Figure~\ref{fig:locality:xput}.  Performance deteriorates as the percentage of local accesses decreases. The maximum achievable throughput  drops slightly, from 350 to 300 KTx/sec. Latency is more penalized, increasing from  8 ms to 150 ms. 
We note that the number of threads needed to saturate the system increases as the locality decreases (from 32 to 512 in this case), because requests spend much of their times traveling across DCs. This explains why the maximum throughput decreases only by 16\% as opposed to the order-of-magnitude increase in latency. 

As any partially replicated system, \ts{} targets workloads with high locality in the data access pattern. In case of limited locality, the performance penalty incurred by \ts{}, and partial replication in general, is the inevitable price to pay to enable higher storage capacity.

\subsection{Data staleness}
\label{sec:eval:visibility}
We measure the staleness of the data returned by \ts{} by measuring the {\em visibility latency} of updates. The visibility latency of an update $X$ in $DC_i$ is the difference between the wall-clock time when $X$ becomes visible in $DC_i$ and the  wall-clock time when $X$ was committed in its original DC.  Figure~\ref{fig:vis_lat} shows the CDF of the update visibility latency achieved by \ts{} and BPR  with 5 DCs and the default workload. 
The CDFs are computed as follows: we first obtain the CDF  on every partition and then we compute the mean for each percentile. 

BPR achieves lower update visibility latency than \ts{}. The update visibility time in \ts{} is higher than in BPR (with an around 200 ms difference in  the worst case). That is to be expected because UST identifies a lower bound of the update times of transactions applied in the whole cluster. BPR effectively trades data freshness for performance, because it exposes more recent snapshots of the data at the cost of blocking reads, hence achieving much lower performance than \ts{}.

\begin{figure}[]
 \vspace{-5mm}
 \begin{center}
  \includegraphics[scale = 0.4]{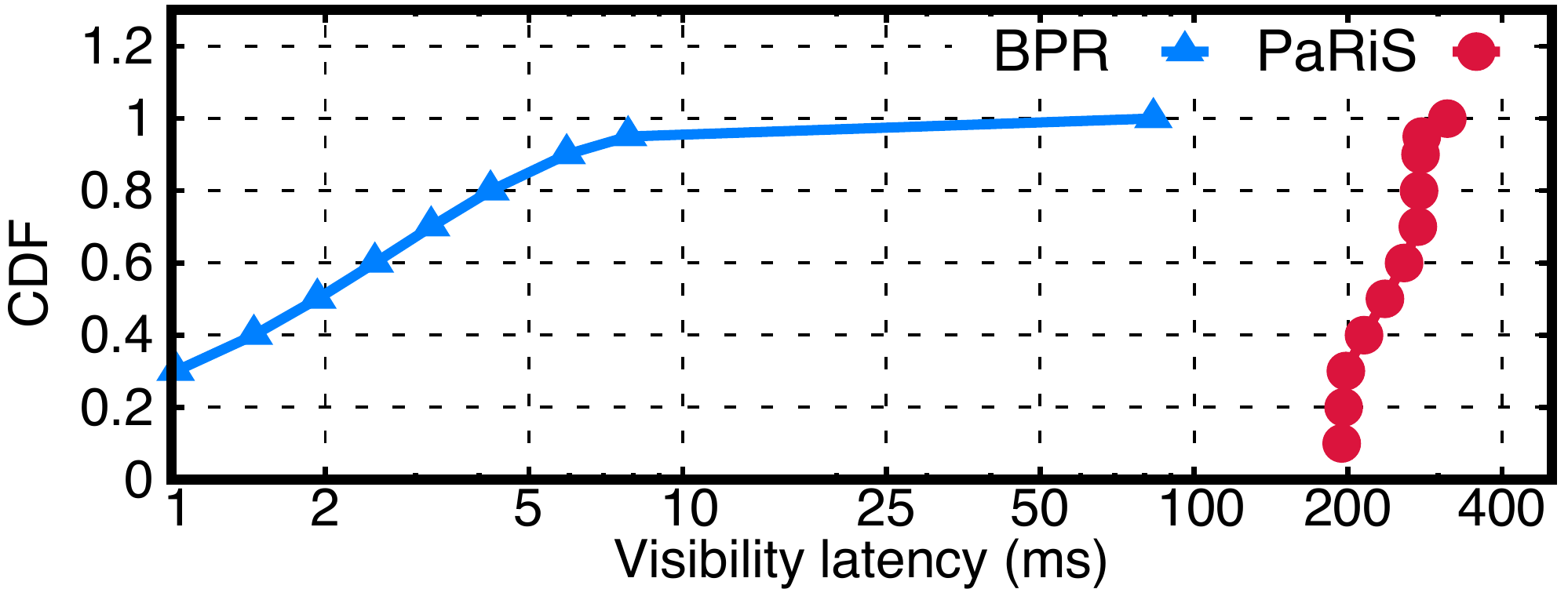}
  \end{center}
   \caption{ \ts{} has higher update visibility latency than BPR (logarithmic scale).}
      \label{fig:vis_lat}
\end{figure}

%% file: sections/related_work.tex
\section{Related work}
\label{sec:relwork}
Table~\ref{tab:rw} classifies existing CC systems according to the transactional model they expose, the capability of achieving non-blocking parallel reads, support for partial replication, and meta-data requirements. The table focuses on systems that target the system model described in Section~\ref{sec:model:model}. 

The vast majority of the systems assume full replication and provide none or restricted transactional capabilities. This class of systems includes COPS~\cite{Lloyd:2011}, Eiger~\cite{Lloyd:2013}, ChainReaction~\cite{Almeida:2013}, Orbe~\cite{Du:2013}, GentleRain~\cite{Du:2014}, Bolt-on CC~\cite{Bailis:2013}, Contrarian~\cite{Didona:2018}, POCC~\cite{Spirovska:2017}, CausalSpartan~\cite{Roohitavaf:2017}, COPS-SNOW~\cite{Lu:2016} and EunomiaKV~\cite{Gunawardhana:2017}. All these systems implement one-shot read-only transactions, and only Eiger additionally supports one-shot write-only transactions.

A few systems support partial replication, i.e., Saturn~\cite{Bravo:2017}, C$^3$~\cite{Fouto:2018}, Karma~\cite{Mahmood:2018} and the one by Xiang and Vaidya~\cite{Xiang:2018}. These systems, however, implement only single-object read and write operations. Among them, only Karma discusses extensions to support read-only transactions by using an approach similar to Eiger's.

To the best of our knowledge, only four systems implement TCC. Among these, OCCULT\footnote{OCCULT may retry read operations multiple times, instead of blocking the read. Retrying has the same effect on latency of blocking the read until the correct version to read is available on the sever that processes the operation.}~\cite{Mehdi:2017}  and Cure~\cite{Akkoorath:2016} can block reads on a node waiting for a snapshot to be installed on such node. Wren~\cite{Spirovska:2018} and AV~\cite{Tomsic:2018} avoid blocking by identifying stable snapshots in a way that is similar to \ts{}. All these systems, however, target full replication.

\begin{table}[]
 \vspace{-5mm}
\scriptsize
\begin{tabular}{|c|c|c|c|c|}
\hline
\textbf{System}        & \textbf{Txs} & \textbf{Nonbl. reads} & \textbf{Partial rep.} & \textbf{Meta-data} \\ \hline\hline
COPS~\cite{Lloyd:2011}                   & ROT                   & $\checkmark$                    & \xmark                     & O$|$deps$|$            \\ \hline

Eiger~\cite{Lloyd:2013}                  & ROT/WOT             & $\checkmark$                    & \xmark                     & O$|$deps$|$            \\ \hline
ChainReaction~\cite{Almeida:2013}          & ROT                   & \xmark                   & \xmark                     & M                  \\ \hline

Orbe~\cite{Du:2013}                   & ROT                   & \xmark                   & \xmark                     & 1 ts      \\ \hline
Gentlerain~\cite{Du:2014}             & ROT                   & \xmark                   & \xmark                     & 1 ts        \\ \hline

POCC~\cite{Spirovska:2017}                   & ROT                   & \xmark                   & \xmark                     & M                  \\ \hline
COPS-SNOW~\cite{Lu:2016}              & ROT                   & $\checkmark$                   & \xmark                     & O$|$deps$|$            \\ \hline

OCCULT~\cite{Mehdi:2017}                 & Generic               & \xmark              & \xmark                     & O(M)               \\ \hline

Cure~\cite{Akkoorath:2016}                   & Generic               & \xmark                   & \xmark                     & M                  \\ \hline
Wren~\cite{Spirovska:2018}                   & Generic               & $\checkmark$                   & \xmark                     & 2 ts                \\ \hline

AV~\cite{Tomsic:2018}          & Generic               & $\checkmark$                   & \xmark                     & M                  \\ \hline
Xiang, Vaidya~\cite{Xiang:2018}       & \xmark              & \xmark                   & $\checkmark$                     & 1 ts               \\ \hline

Contrarian~\cite{Didona:2018}             & ROT                   & $\checkmark$                   & \xmark                     & M                  \\ \hline
C$^3$ ~\cite{Fouto:2018}                 & \xmark              & {$\checkmark$}                          & $\checkmark$                     & {M}                  \\ \hline

Saturn~\cite{Bravo:2017}                 & \xmark              & $\checkmark$                   & $\checkmark$                     & 1 ts                 \\ \hline
Karma~\cite{Mahmood:2018}                  & ROT                   & {$\checkmark$}                          & $\checkmark$                     & {O$|$deps$|$}                  \\ \hline

CausalSpartan~\cite{Roohitavaf:2017}          & \xmark              & $\checkmark$                   & \xmark                     & M                  \\ \hline
Bolt-on CC~\cite{Bailis:2013}        & \xmark              & $\checkmark$                   & \xmark                     & M                  \\ \hline
EunomiaKV~\cite{Gunawardhana:2017}          & \xmark              & $\checkmark$                   & \xmark                     & M                  \\ \hline\hline
{\bf \ts{}} (this work) & Generic               & $\checkmark$                   & $\checkmark$                     & 1  ts               \\ \hline
\end{tabular}
\caption{Taxonomy of the main CC systems. $M$ is the number of DCs. $ts$ stands for timestamp. For systems that do not support transactions, the {\em non-blocking read} property refers to single-item reads. \ts{} is the only system that supports partial replication with generic transactions, non-blocking parallel reads, and constant meta-data to track dependencies. }
\label{tab:rw}
\end{table}

Other relevant systems include  TARDiS~\cite{Crooks:2016b}, GSP~\cite{gsp},  SwiftCloud~\cite{Zawirski:2015}, Lazy Replication~\cite{Ladin:1992}, ISIS~\cite{Birman:1991} and Bayou~\cite{Terry:1995}. These systems do not support sharding, and hence neither partial replication. Many protocols have also been proposed to implement causally consistent distributed shared memories, e.g.,~\cite{ Baldoni:2006,Xiang:2018b,Hsu:2018}. These protocols do not support transactions and require more meta-data than \ts{}.

\ts{} is also related to systems that implement stronger consistency levels and support partial replication, such as Jessy~\cite{Ardekani:2013}, P-Store~\cite{Schiper:2014}, STR~\cite{Li:2018}, and Spanner~\cite{Corbett:2013}. On the one hand, these systems allow fewer anomalies than what is allowed by TCC~\cite{Akkoorath:2016} and provide fresher data to the clients. On the other hand, they incur higher synchronization costs to determine the outcome of transactions. \ts{} targets applications that can tolerate weaker consistency and some degree of data staleness, e.g., social networks, and offers them low latency, scalability and high storage capacity. 

%% file: sections/conclusion.tex
\section{Conclusion}
\label{sec:conclusion}
We present \ts{}, the first system that implements TCC in a partially replicated system and achieves non-blocking read operations. \ts{} implements a novel  dependency tracking protocol, called UST, which requires only one timestamp to track dependencies. UST identifies a snapshot of the data that is available at {\em every} DC, thereby enabling non-blocking reads regardless of the DC in which the read takes place.

We evaluate \ts{} on a data platform replicated on up to 10 DCs. \ts{} scales well and achieves lower latency than the blocking alternative, while being able to handle larger datasets than existing solutions that assume full replication.